\documentclass[a4paper,USenglish]{lipics-v2021}
\usepackage{amsmath,amsthm,amssymb}
\usepackage{microtype}
\usepackage{apxproof}
\usepackage{hyperref}

\usepackage{graphicx}
\usepackage{verified-badges}
\newcommand{\leanurl}[1]{https://github.com/BioDisCo/roots/blob/master/lean/AsymptoticSubspace/#1}
\newcommand{\suppurl}[1]{https://github.com/BioDisCo/roots/blob/master/manuscript/supplementary_material/#1}
\usepackage{xcolor}
\definecolor{suppcolor}{HTML}{1A6496}
\newcommand{\suppref}[2]{\href{\suppurl{#1}}{\textcolor{suppcolor}{\textbf{#2}}}}

\usepackage{tikz}
\usetikzlibrary{shapes.symbols,arrows}
\usetikzlibrary{matrix,calc,intersections,patterns}

\usepackage{algorithm,algorithmic}
\usepackage{pgfplots}
\pgfplotsset{compat=1.18}

\newcommand{\INITIALLY}{\REQUIRE{}}
\newcommand{\ROUND}{\ENSURE{}}
\newcommand{\bra}{\left\langle}
\newcommand{\ket}{\right\rangle}

\newcommand{\nop}[1]{}

\newtheoremrep{theorem}{Theorem}
\newtheoremrep{lemma}[theorem]{Lemma}
\newtheoremrep{corollary}[theorem]{Corollary}

\DeclareMathOperator{\vol}{vol}

\DeclareMathOperator{\dist}{dist}
\DeclareMathOperator{\linspan}{span}
\DeclareMathOperator{\dir}{dir}

\DeclareMathOperator{\Sym}{Sym}

\newcommand{\N}{\mathcal{G}}

\newcommand{\IR}{\mathbb{R}}
\newcommand{\IN}{\mathbb{N}}

\newcommand{\Ball}{B}

\renewcommand{\le}{\leqslant}
\renewcommand{\leq}{\leqslant}
\renewcommand{\ge}{\geqslant}
\renewcommand{\geq}{\geqslant}

\usepackage{graphicx}

\DeclareMathOperator\In{In}

\DeclareMathOperator\poly{poly}

\title{Asymptotic Subspace Consensus in Dynamic Networks}

\author{Matthias {Függer}}{Université Paris-Saclay, CNRS, ENS Paris-Saclay, LMF, Gif-sur-Yvette, France \and \url{https://home.lmf.cnrs.fr/MatthiasFuegger/} }{mfuegger@lmf.cnrs.fr}{https://orcid.org/0000-0001-5765-0301}{}

\author{Thomas {Nowak}}{Université Paris-Saclay, CNRS, ENS Paris-Saclay, LMF, Gif-sur-Yvette, France \and Institut Universitaire de France, Paris, France \and \url{https://www.thomasnowak.net}}{thomas@thomasnowak.net}{https://orcid.org/0000-0003-1690-9342}{}

\authorrunning{M. Függer and T. Nowak} 

\funding{The work was supported by the French National Research Agency
(ANR) projects DREAMY (ANR-21-CE48-0003) and COSTXPRESS (ANR-23-CE45-0013), as well as the SAIF project, funded by the ``France 2030'' government investment plan managed by ANR, under the reference ANR-23-PEIA-0006.}

\acknowledgements{We thank Emmanuel Godard and Eloi Perdereau for introducing us to the question of possible generalizations of the asymptotic consensus problem.}

\nolinenumbers 
\hideLIPIcs 

\Copyright{M. Függer and T. Nowak} 



\supplement{\url{https://github.com/BioDisCo/roots}}

\begin{document}

\maketitle

\begin{abstract}
	We introduce the problem of asymptotic subspace consensus, which requires the outputs of processes to converge onto a common subspace while remaining inside the convex hull of initial vectors.
	This is a relaxation of asymptotic consensus in which outputs have to converge to a single point, i.e., a zero-dimensional affine subspace.

	We give a complete characterization of the solvability of asymptotic subspace consensus in oblivious message adversaries.
	In particular, we show that a large class of algorithms used for asymptotic consensus gracefully degrades to asymptotic subspace consensus in distributed systems with weaker assumptions on the communication network.
	We also present bounds on the rate by which a lower-than-initial dimension is reached.
\end{abstract}

\section{Introduction}

When computation is distributed across several processes, an often encountered problem is the one of reaching consensus on a common property among the processes.
Consensus problems have thus been extensively studied in many variants.
Examples include agreement on a discrete state like a transaction status \cite{Lamport82,schneider1990implementing}, on continuous values such as time \cite{Hal84,srikanth87,LL88,dolev2004self,kopetz2009clock}, and on multidimensional properties like three-dimensional coordinates in space \cite{Bal05,murray2007recent,mendes15:multidimensional,FN18:disc}.

This work falls into the latter category of agreement on multidimensional continuous values.
We assume a round-wise computational model in which processes communicate and update their state in rounds.
In each round a process sends messages, receives messages, and updates its local state according to the previous state and received messages.
The dynamics of the underlying communication network are modeled by a set of communication graphs that may change arbitrarily in each round.
So-called oblivious message adversaries have been widely used in distributed computing to model highly dynamic networks \cite{SW89,charron09:HO,afek13:asynchrony,CG15}.

In this model the well-studied problem of asymptotic consensus is formulated as: Each process $i$ starts with an initial input $x_i(0) \in \IR^d$, for some dimension $d \geq 1$.
An algorithm solves asymptotic consensus if all process outputs $x_i(t)$ converge to a common value (Agreement property) that lies within the convex hull of initial values (Validity property).
Closely related to this problem is the problem of $\varepsilon$-consensus, or approximate consensus, where processes terminate and (Agreement) is replaced with the requirement that values in the final round are within distance $\varepsilon > 0$.

Charron-Bost, Függer, and Nowak \cite{CBFN15:icalp} gave a characterization of the oblivious message adversaries in which one-dimensional asymptotic consensus is solvable.
The authors showed that simple averaging algorithms with mild conditions on their weights solve the problem whenever it is solvable and that the latter is the case exactly in \emph{rooted oblivious message adversaries}, i.e., where each communication graph has at least one process (the root of this graph) that reaches all other processes.
Subsequent work showed that alternating flooding and averaging considerably speeds up convergence \cite{CBFN16:icalp}.
Algorithms based on this scheme were shown to be asymptotically optimal by F{\"u}gger, Nowak, and Schwarz \cite{fugger2021tight}.

Averaging algorithms naturally generalize to dimensions $d > 1$, with the same characterization of solvability for asymptotic and $\varepsilon$-consensus, and the main challenge being high convergence rates despite large dimensions~\cite{CBFN16:centroid,FN18:disc}.
Multidimensional $\varepsilon$-consensus has also been studied in the context of adversarial faults and asynchrony by Mendes, Herlihy, Vaidya, and Garg \cite{mendes15:multidimensional} with a time to reach $\varepsilon$-consensus that depends linearly on~$d$.
Via a reduction to dynamic networks, Függer and Nowak \cite{FN18:disc} showed that this dependency can be dropped, with the convergence rate of asymptotic consensus and the termination time in $\varepsilon$-consensus being independent of~$d$ for certain averaging algorithms.
Approximate consensus was also investigated in non-synchronous settings \cite{attiya2023step,ghinea2023multidimensional}.

In this work we investigate the behavior of averaging algorithms in weaker than rooted oblivious message adversaries,
where asymptotic consensus is not necessarily solvable.
This is motivated by an understanding of how averaging algorithms degrade in non-rooted oblivious message adversaries,
and by the hypothesis that some applications may not require convergence to a single value.

\paragraph*{Results and outline}

We define the problem of \emph{$d$-to-$s$-dimensional asymptotic subspace consensus}, with $d > s$, by weakening
(Agreement) into (Subspace Agreement), requiring that the limits of the processes outputs are within an $s$-dimensional subspace
of the $d$-dimensional space.
We show that $d$-to-$s$-dimensional asymptotic subspace consensus is solvable if and only if the oblivious message adversary is \emph{$(s+1)$-rooted}, i.e., in each graph, each process is reachable from at least one of a set of $s+1$ roots.
In case it is solvable, averaging algorithms solve it under weak assumptions on their weights.
In particular, this is fulfilled when choosing equal weights (Figure~\ref{fig:panel_1}a).
Figures~\ref{fig:panel_1}b--d show simulations of averaging in an initially 3-dimensional space that are seen to converge onto 0-dimensional in 1-rooted (b), 1-dimensional in 2-rooted (c), and 2-dimensional subspaces in 3-rooted oblivious message adversaries (d).
We prove that this in the general case as outlined:

We introduce the model and the problem in Sections~\ref{sec:model} and~\ref{sec:problem}, respectively.
Our analysis builds upon a reduction from $k$-rooted oblivious message adversaries to \emph{$k$-broadcastable} adversaries, where each process directly receives a message from one of $k$ processes (called the \emph{broadcasting set}).
Asymptotically optimal bounds for the complexity of this reduction have been proven by El-Hayek, Henzinger, and Schmid~\cite{el2023asymptotically}.

Section~\ref{sec:problem} also establishes a lower bound for the message adversaries in which the problem of asymptotic subspace consensus can be solved.
Section~\ref{sec:algos} describes a class of algorithms that allow to solve asymptotic subspace consensus, namely averaging algorithms.
In Section~\ref{sec:analysis} we prove our two main results:
In Section~\ref{sec:converge_zero_vol}, a rate at which certain averaging algorithms contract to a space that is of lower dimension than the space of the initial values.
The proof is by showing that the volume that contains the process outputs converges to~0.
Inspired by the symmetrization of Charron-Bost, Függer, and Nowak~\cite{CBFN16:centroid}, we use a Steiner-type symmetrization of the convex hull of the process outputs.
Most importantly, this symmetrization preserves volumes of cuts along the first axis and guarantees concavity of the ball-radius function along the first axis.
We then establish a lower bound on the volume contraction between successive rounds.

Section~\ref{sec:converge_lowdim} presents a complete characterization of oblivious message adversaries with respect to the achievable reduction in dimensionality.
The proof is based on an observation shown in Figures~\ref{fig:panel_1}e at the example of $2$-rooted message adversaries.
Tracking the dynamics of the affine subspace spanned by the outputs of two processes that are in the broadcasting set, one observes convergence of this subspace and attraction of the other process outputs to this space.
Again, we show that this holds in the general case.

We conclude in Section~\ref{sec:conclusion}.

Some definitions and results have been formalized in the Lean~4 proof assistant~\cite{lean4} using the Mathlib library.
Results marked with \leanproof{} have machine-checked proofs; definitions and statements marked with \leanformalized{} have been formalized without a complete proof (see \url{https://github.com/BioDisCo/verified-badges}).
The Lean formalizations and proofs were generated with the assistance of Claude Opus~4.6 and GPT-5.3-Codex.
These formalizations serve as supplementary information and are not meant as replacements for the natural-language proofs in this paper.
Clicking a badge links to the corresponding Lean source code.

\begin{figure}[hpt]
	\centering
	\includegraphics[width=\linewidth]{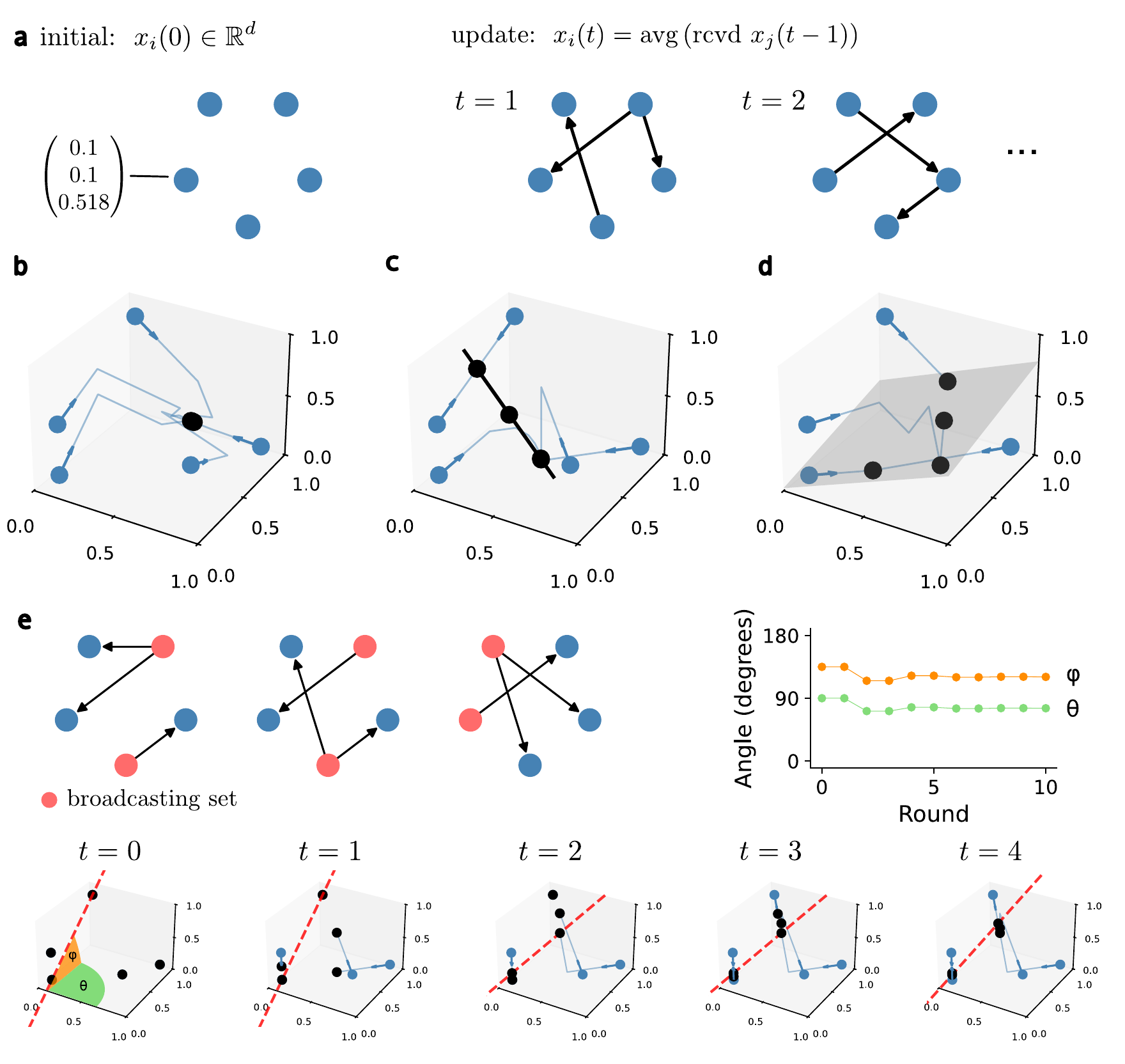}
	\caption{\textbf{Averaging algorithm running in dynamic networks.} \textbf{a} Distributed system with $5$ nodes, executing an averaging algorithm in $\IR^3$. Nodes start with initial values in $\IR^3$ and average values received within the round's communication graph ($2$-rooted graphs shown). \textbf{b--d} Execution of the equal neighbor algorithm in $\IR^3$ for $10$ rounds. Initial values (blue) and round-10 values (black) shown. \textbf{b} In a $1$-rooted oblivious message adversary, convergence onto a single point (black).  \textbf{c} In a $2$-rooted oblivious message adversary, convergence onto a line (black) \textbf{d} In a $3$-rooted oblivious message adversary, convergence onto a plane (black).
		Animated (Movies~\suppref{Movie_S1_execution_1_root_animation.mp4}{S1}, \suppref{Movie_S2_execution_2_roots_animation.mp4}{S2}, \suppref{Movie_S3_execution_3_roots_animation.mp4}{S3}) and interactive (Documents~\suppref{Document_S1_execution_1_root_execution.html}{S1}, \suppref{Document_S2_execution_2_roots_execution.html}{S2}, \suppref{Document_S3_execution_3_roots_execution.html}{S3}) versions of \textbf{b--d} are provided as supplementary material.
		\textbf{e} Simulation of averaging algorithm within an oblivious message adversary with $2$-broadcastable communication graphs
		(broadcasting set marked in red).
		Initial values ($t=0$) and outputs until $t=4$ are shown (black).
		The linear subspace spanned by the processes in the current round's broadcasting set is shown (dashed red).
		Its polar coordinates ($\varphi$ and $\theta$) are seen to converge.
		An animated execution oscillating on a $1$-dimensional subspace is provided as Movie~\suppref{Movie_S4_execution_oscillating_on_1d_animation.mp4}{S4} and Document~\suppref{Document_S4_execution_oscillating_on_1d_execution.html}{S4}.
	}
	\label{fig:panel_1}
\end{figure}

\section{Computational Model\leanformalized{\leanurl{DefComputationalModel.lean}}}\label{sec:model}

We write $\IN = \{1,2,\dots\}$ and $\IN_0 = \IN \cup \{0\}$.
Further, $[n] = \{1,\dots, n\}$.
For directed graphs $G_1 = ([n],E_1)$ and $G_2 = ([n],E_2)$, we write
$G_1 \circ G_2$ for the product graph $G = ([n], E)$ with $(i,j) \in E$
if and only if there exists a $u \in [n]$ such that $(i,u) \in E_1$ and
$(u,j) \in E_2$.

We assume a system of $n > 1$ processes that communicate via a synchronous message-passing network.
In each round, processes update their state deterministically based on their local state and the
messages they receive.
Part of their local state is an output, with the output of process $i \in [n]$ in round $t \geq 0$ denoted as $x_i(t)$.
In this work we assume that a process output is within some $\IR^d$.

Network links are dynamic, akin to message adversaries~\cite{afek13:asynchrony} or the Heard-Of
model~\cite{charron09:HO}.
That is, every process $i\in[n]$ sends a message in every round $t\in\IN$,
then receives the messages from all incoming neighbors $\In_i(t)$ in the (directed)
round-$t$ communication graph $G_t=([n],E_t)$, and updates its local state as a function
of its previous state and the received messages.
Every process receives its own message, i.e., $(i,i)\in G_t$ for all $i\in[n]$.
The execution of a deterministic algorithm is uniquely determined by the
initial states and the sequence of communication graphs.
Although the model is formally synchronous, it has been shown to capture
classical round-based asynchronous~\cite{charron09:HO} and even non-benign~\cite{biely07:tolerating,mendes15:multidimensional} models.

\paragraph*{Message adversaries}

The dynamics of communication links is described by message adversaries.
In this work we focus on dynamic networks that can switch arbitrarily between communication graphs from a so-called \emph{oblivious message adversary}~$\N$, which is a non-empty set of communication graphs.
We say an algorithm \emph{solves a problem} in~$\N$ if it satisfies its specification in all executions in which all communication graphs are chosen from~$\N$.

For many non-trivial problems, there is no algorithm that solves the problem in an arbitrary~$\N$.
For example, Charron-Bost, Függer, and Nowak~\cite{CBFN15:icalp} have shown that asymptotic consensus,
i.e., 1-to-0-dimensional asymptotic subspace consensus, is
solvable in precisely those $\N$ that contain only rooted graphs.
The following definitions are natural generalizations of this concept.

A communication graph~$G$ is \emph{$k$-rooted} if there exists a set
$M\subseteq[n]$ of at most~$k$ processes such that every process in~$[n]$
is reachable from~$M$ in~$G$.
In this case, we call~$M$ a \emph{root set} of~$G$.
An oblivious message adversary is $k$-rooted if all its communication graphs are.

Likewise, graphs that contain a star play a central role in consensus-type problems~\cite{charron09:HO}.
Generalizing such graphs, we define:
A communication graph~$G$ is \emph{$k$-broadcastable} if there exists a set
$M\subseteq[n]$ of at most~$k$ processes such that every process in~$[n]$
has an incoming edge from at least one process in~$M$ in~$G$.
In this case, we call~$M$ a \emph{broadcasting set} of~$G$.
Every $k$-broadcastable communication graph is $k$-rooted, but the converse is
not true.
An oblivious message adversary is $k$-broadcastable if all its communication graphs are.

We next discuss links between these two classes of graphs.

\paragraph*{From Rooted to Broadcastable Graphs}

Relaying messages over multiple rounds with communication graphs $G_1$, \dots, $G_t$ results in reception of the message according to the product graph $G = G_1 \circ\cdots\circ G_t$ in round~$t$.
This fact can be used in reductions from one adversary to another by simulating a round with communication graph $G$ from rounds with graphs $G_1$, \dots, $G_t$.

Charron-Bost, Függer, and Nowak~\cite{CBFN15:icalp} showed a reduction from any sequence of $n-1$ 1-rooted communication graphs to so-called non-split graphs, of which any sequence of $O(\log\log n)$ graphs results in a 1-broadcastable graph \cite{fugger2020radius}, i.e., graphs that contain a star.
The resulting bound of $O(n \log\log n)$ rounds from 1-rooted to 1-broadcastable was later improved by El-Hayek, Henzinger, and Schmid~\cite{el2023asymptotically} to $\lceil(1+\sqrt{2}n)-1\rceil$ graphs.

For the more general $k$-broadcastable graphs, one can show that from $k$-rooted communication graphs, one also obtains $k$-broadcastable graphs in a bounded number of rounds.
A short argument for a polynomial bound is as follows:

Let $t \geq n^{k+1}$ and let $G_1,\dots,G_t$ be $k$-rooted.
By the pigeonhole principle, there exists a subsequence $G_1',\dots,G_s'$ of at least
$s \geq n^{k+1}/\binom{n}{k}\geq n^{k+1}/n^k = n$ communication graphs in which~$M$ is a root set.
By considering the influence sets
$S_i(r) = \big\{ j\in[n] \mid (i,j) \in E(G_1'\circ \cdots \circ G_r') \big\}$ for every $i\in M$ and
$S(r) = \bigcup_{i \in M} S_i(r)$,
one observes that $S(r)$ strictly grows until it covers the set~$[n]$ of all processes.
Thus, since $\lvert S(0)\rvert \leq k$ and $s \geq n$, influence set $S(s)$ covers~$[n]$.
In particular this means that~$M$ is a broadcasting set in the product graph $G = G_1\circ\cdots\circ G_t$.

A more involved analysis by El-Hayek, Henzinger, and Schmid~\cite{el2023asymptotically} shows that indeed a bound linear in $n$, and independent of $k$, holds.

\begin{theorem}[{\cite[Theorem 32]{{el2023asymptotically}}}]\label{thm:rooted:to:broadcastable}
	Any product of at least~$\frac{\pi^2+6}{6}n + 1$ communication graphs that are $k$-rooted is
	$k$-broadcastable.
\end{theorem}

\section{Asymptotic Subspace Consensus}\label{sec:problem}

We start with basic geometric notations.
For a set $A\subseteq \IR^d$, we denote by $\bar{A}$ the topological closure
of~$A$.
For any non-empty finite set $X\subseteq \IR^d$, denote by $\poly(X)$ the
polyhedron generated by~$X$, i.e., its convex hull.
For any (Lebesgue) measurable set $A\subseteq \IR^d$, denote its volume in $\IR^d$ by
$\vol_d(A)$, or simply $\vol(A)$ if the dimension is clear from the context.
We write $\lVert z \rVert = \sqrt{\sum_{k=1}^d z_k^2}$ for the Euclidean norm of any $z\in\IR^d$.
For sets $A, B\subseteq \IR^d$, denote by
$\dist(A,B)
	=
	\inf_{a \in A, b \in B} \lVert a - b \rVert
$
their Euclidean distance.
We note that $\dist$ is not a metric; it does not satisfy the triangle inequality and $\dist(A,B) = 0$ does not necessarily imply $A=B$.

We say that a sequence of vectors $x(t) \in \IR^d$, $t \in \IN$, \emph{converges onto}
a set $X \subseteq \IR^d$ if $\displaystyle\lim_{t\to\infty}\dist(x(t),X) = 0$.
Equipped with this, we are now in the position to state the asymptotic subspace consensus problem:

In the \emph{$d$-to-$s$-dimensional asymptotic subspace consensus problem}, with $d > s$, every process~$i \in [n]$ starts with an initial vector $x_i(0) \in \IR^d$ and outputs a vector $x_i(t) \in \IR^d$ in every round $t\in\IN$ such that in every execution:
\begin{description}
	\item[(Subspace Agreement)] There exists an affine subspace $E\subseteq \IR^d$ of dimension~$s$, such that all sequences $\big(x_i(t)\big)_{t \geq 0}$ converge onto~$E$.

	\item[(Validity)] All sequences $\big(x_i(t)\big)_{t \geq 0}$ converge onto the convex hull of the set of initial vectors.

\end{description}

Asymptotic consensus on scalar inputs is a special case with $d=1$ and $s=0$.

\subsection{Lower bound for dimension reduction}

We first show a lower bound on the attainable dimensionality reduction in asymptotic subspace consensus in terms of the sizes of the root sets of oblivious message adversaries.
The proof is a generalization of the impossibility of asymptotic consensus in non-rooted oblivious message adversaries \cite{CBFN15:icalp}.
It does not rely on the dynamics of the communication network and requires a single static communication graph only, where a sufficiently large number of processes are isolated and remain on their initial values.
The proof is given in the Appendix.

\begin{theoremrep}[\leanproof{\leanurl{ThmImposs.lean}}]\label{lem:imposs}
	For $d > s \geq 0$, the $d$-to-$s$-dimensional asymptotic subspace consensus problem is unsolvable in oblivious message adversaries that are not $(s+1)$-rooted.
\end{theoremrep}
\begin{proof}
	In the following fix an $s \geq 0$ and set $k = s+1$.
	If an oblivious message adversary is not $k$-rooted, then it contains a communication
	graph~$G$ that is not $k$-rooted.
	But then~$G$ contains at least $k+1$ strongly connected components without incoming edges.
	Let~$M$ be a set of $k+1$ processes from different such components.
	Without loss of generality, by permuting the process names, assume that $M =
		\{1,\dots,k+1\}$.
	For every $i\in M$, let~$A_i$ be the set of processes that reach~$i$ in~$G$.
	By choice of~$M$, the sets~$A_i$ are pairwise disjoint.

	We will show by means of contradiction that for $d = k$ there is no algorithm that solves $d$-to-$s$-dimensional asymptotic subspace consensus in oblivious message adversary $\N = \{G\}$.
	The theorem's statement then holds for $d' > d$ by embedding $\IR^d$ in $\IR^{d'}$.

	We construct an execution of the algorithm as follows.
	For every $j\in A_i$, choose the initial vectors $x_j(0) = e_i \in \IR^d$ if $1\leq i \leq k$ where~$e_i$ is the $i$\textsuperscript{th} standard basis vector, and $x_j(0) = 0$ if $i = k+1$.
	Choose the communication graph of every round~$t$ to be $G_t = G$.
	By Validity, for every $i\in M$, since the convex hull of initial values
	in~$A_i$ is a singleton and~$i$ does not hear from any process outside
	of~$A_i$, the sequence $\big(x_i(t)\big)_t$ converges onto $\{x_i(0)\}$, that is, it
	converges to~$x_i(0)$.
	But the $k+1 = s+2$ initial values, and thus the limits, of processes in~$M$ do not
	belong to a common subspace of dimension~$s$.
	This is a contradiction to (Subspace Agreement).
\end{proof}

\section{Averaging Algorithms}\label{sec:algos}

A deceptively simple class of algorithms are averaging algorithms.
An \emph{averaging algorithm} keeps only its current vector~$x_i(t)$ as its
state, which it sends in every round.
After having received the vectors of other processes, it updates its vectors to
a weighted average of the received vectors, see Algorithm~\ref{alg:averaging}.
In concordance with the term averaging, the weights~$w_{ij}(t)$ are assumed to be non-negative and $\sum_{j\in [n]} w_{ij}(t) = 1$.
Averaging algorithms differ in choices of weights.
Maybe the most natural averaging algorithm is the \emph{equal neighbor}
algorithm, which assigns the same weight to all received vectors.

\begin{algorithm}[ht]
	\begin{algorithmic}[1]
		\INITIALLY{}
		\STATE $x_i(0)$ is process $i$'s initial vector in~$\IR^d$
		\ROUND{}
		\STATE send $x_i(t-1)$
		\STATE receive vectors $x_j(t-1)$ for $j\in \In_i(t)$
		\STATE determine weights~$w_{ij}(t)$ for received vectors
		\STATE
		$\displaystyle x_i(t) \gets \sum_{j\in \In_i(t)} w_{ij}(t) x_j(t-1)$
	\end{algorithmic}
	\caption{Averaging algorithm for process $i$}
	\label{alg:averaging}
\end{algorithm}

For ease of notation, we set $w_{ij}(t) = 0$ if $j \not\in \In_i(t)$, allowing us to write  the update step as $x_i(t) = \sum_{j\in [n]} w_{ij}(t) x_j(t-1)$.

Following Charron-Bost, Függer, and Nowak~\cite{CBFN16:icalp}, we call an averaging algorithm \emph{$\alpha$-safe} if it guarantees a minimum positive weight $w_{ij}(t)\geq \alpha > 0$ for messages it receives, i.e., from some $j\in \In_i(t)$.
For example, the equal neighbor algorithm is $1/n$-safe.

We next define a property for averaging algorithms executed in $k$-broadcastable oblivious message adversaries for some $k \geq 1$.
The property will be central in quantifying the influence a broadcasting set has on the other processes.

\begin{definition}[Minimum broadcasting weight\leanformalized{\leanurl{DefMinimumBroadcastWeight.lean}}]
	Let round $t \geq 1$, graph $G_t$ be $k$-broadcastable for some $k \geq 1$, and $M(t)$ be the broadcasting set in $G_t$.
	For an averaging algorithm, an update step $x_i(t) = \sum_{j \in [n]} w_{ij}(t) x_j(t-1)$ has minimum broadcasting weight $\alpha$ if $\sum_{j \in M(t)} w_{ij}(t) \geq \alpha$.
	An averaging algorithm has minimum broadcasting weight $\alpha$ if all update steps do.
\end{definition}

Executed in a $k$-broadcastable oblivious message adversary, an averaging algorithm that is $\alpha$-safe has minimum broadcasting weight $\alpha$, but not necessarily vice versa.

\section{Analysis}\label{sec:analysis}

We start the analysis with some notation.
Denote by~$X(t)$ the set
$\{ x_i(t) \mid i \in [n] \}$
of vectors in round~$t$
and
by
$P(t) = \poly(X(t))$ the generated polyhedron.
For a set $A \subseteq [n]$, we
abbreviate $X_A(t) = \{ x_i(t) \mid i \in A \}$
as well as $P_A(t) = \poly(X_A(t))$.

In the following Section~\ref{sec:distance_H} we establish a lower bound on the attraction a broadcasting set has on the other processes.

\subsection{Distance to Hyperplane}
\label{sec:distance_H}

The following is a formula for the distance of a point to a half-space.
Its proof is given in the Appendix for completeness.

\begin{lemmarep}[\leanproof{\leanurl{LemHalfspaceDistFormula.lean}}]\label{lem:halfspace:distance:formula}
	Let point $q\in \IR^d$ and normal vector $v\in \IR^d\setminus\{0\}$ define the half-space $H = \{ h\in\IR^d \mid \langle h - q , v \rangle < 0 \}$.
	Then
	$\dist(\{z\}, H) = \langle z - q , v\rangle / \lVert v\rVert$
	for all points $z \in \IR^d \setminus H$.
\end{lemmarep}
\begin{proof}
	Let \(z\in\mathbb{R}^d\setminus H\). Then \(\langle z-q,v\rangle\ge 0\).
	For any \(h\in H\) we have \(\langle h-q,v\rangle<0\), hence
	\(
	\langle z-h,v\rangle
	= \langle z-q,v\rangle-\langle h-q,v\rangle
	> \langle z-q,v\rangle
	\).
	By the Cauchy--Schwarz inequality,
	\(
	\|z-h\|\,\|v\|\ge \langle z-h,v\rangle > \langle z-q,v\rangle ,
	\)
	and therefore
	\(
	\|z-h\| > \frac{\langle z-q,v\rangle}{\|v\|}
	\).
	Since this holds for all \(h\in H\), it follows that
	\begin{equation}\label{eq:lem:halfspace:distance:formula:lower}
		\dist(\{z\},H)\ge \frac{\langle z-q,v\rangle}{\|v\|}.
	\end{equation}
	Now let \(\varepsilon>0\) and define
	\[
		h_\varepsilon = z-\left(\frac{\langle z-q,v\rangle}{\|v\|^2}+\varepsilon\right)v .
	\]
	Then
	\(
	\langle h_\varepsilon-q,v\rangle
	= \langle z-q,v\rangle
	-\left(\frac{\langle z-q,v\rangle}{\|v\|^2}+\varepsilon\right)\|v\|^2
	= -\varepsilon\|v\|^2<0,
	\)
	so \(h_\varepsilon\in H\). Moreover,
	\[
		\|z-h_\varepsilon\|
		= \Bigl(\frac{\langle z-q,v\rangle}{\|v\|^2}+\varepsilon\Bigr)\|v\|
		= \frac{\langle z-q,v\rangle}{\|v\|}+\varepsilon\|v\|.
	\]
	Taking the infimum over \(H\) and letting \(\varepsilon\to 0\) yields
	\begin{equation}\label{eq:lem:halfspace:distance:formula:upper}
		\dist(\{z\},H)\le \frac{\langle z-q,v\rangle}{\|v\|}.
	\end{equation}
	Combining~\eqref{eq:lem:halfspace:distance:formula:lower} and~\eqref{eq:lem:halfspace:distance:formula:upper} concludes the proof.
\end{proof}

For any averaging algorithm, processes only update their output to a value within the
convex hull, that is, the polyhedron, of received values.
Thus any open half-space $H$ that does not contain points in round $t \ge 0$ will not contain points in round $t+1$.
In fact one can say more for executions in oblivious message adversaries that are $k$-broadcastable, for some $k \geq 1$: processes update their outputs to values that have a minimum distance to such empty half-spaces.
This is due to them being attracted by the values of the broadcasting set.
The following lemma establishes a lower bound on this distance (Figure~\ref{fig:panel_2}a); its proof is given in the Appendix.

\begin{lemmarep}[\leanproof{\leanurl{LemHalfspaceZone.lean}}]\label{lem:halfspace:zone}
	Consider an averaging algorithm executed in some $k$-broadcastable oblivious message adversary with minimum broadcasting weight $\alpha > 0$ and fix a round $t \geq 1$.
	Let $H \subseteq \IR^d$ be an open half-space such that $X(t-1) \cap H = \emptyset$.
	Then $\dist(X(t), H) \geq \alpha \dist(P_{M(t)}(t-1), H)$.
\end{lemmarep}
\begin{proof}
	Let $p \in X_{M(t)}(t-1)$ be a value broadcast in round $t$ and $q\in \bar{H}$ such that $\lVert p - q\rVert = \dist(X_{M(t)}(t-1),\bar{H})$.
	Then $\lVert p - q\rVert = \dist(X_{M(t)}(t-1),H)$.
	Because the lemma is trivial if $\dist(X_{M(t)}(t-1),H)=0$, we assume $p\neq q$ in the rest of the proof.

	Let $i\in[n]$.
	Then, using $\langle x_j(t-1) - q , p - q\rangle \geq 0$ for all $j\in[n]$ because $X(t-1)\cap H = \emptyset$, we have
	\begin{equation}\label{eq:lem:halfspace:zone:ineqs}
		\begin{split}
			\left\langle x_i(t) - q \, , \, p - q \right\rangle
			& =
			\left\langle \sum_{j\in [n]} w_{ij}(t) x_j(t-1) - q \,,\, p - q \right\rangle
			\\ & =
			\left\langle \sum_{j\in [n]} w_{ij}(t) x_j(t-1) - \sum_{j\in[n]} w_{ij}(t) q \,,\, p - q \right\rangle
			\\ & =
			\sum_{j\in [n]} w_{ij}(t) \left\langle x_j(t-1) - q \,,\, p - q \right\rangle\\
			& \geq \sum_{j\in M(t)} w_{ij}(t) \left\langle x_j(t-1) - q \,,\, p - q \right\rangle\\
			& \geq
			\alpha \left\langle \hat{x} - q \,,\, p - q \right\rangle
		\end{split}
	\end{equation}
	for $\hat{x} = 1/\left(\sum_{j\in M(t)} w_{ij}(t) \right) \cdot  \sum_{j\in M(t)} w_{ij}(t) x_j(t-1) \in P_{M(t)}(t-1)$ and by applying
	the minimum broadcasting weight assumption $\sum_{j\in M(t)} w_{ij}(t) \geq \alpha$.

	By Lemma~\ref{lem:halfspace:distance:formula},
	it is
	$\langle \hat{x} - q , p - q \rangle \geq \lVert p - q\rVert\dist(\{\hat{x}\},H) \geq \lVert p - q\rVert\dist(P_{M(t)}(t-1),H)$.
	But then \eqref{eq:lem:halfspace:zone:ineqs} implies
	\begin{equation}\label{eq:lem:halfspace:zone:ineq}
		\left\langle x_i(t) - q \, , \, p - q \right\rangle
		\geq
		\alpha \lVert p - q\rVert \dist(P_{M(t)}(t-1),H)
		\enspace.
	\end{equation}
	Since $x_i(t)\not\in H$, another application of
	Lemma~\ref{lem:halfspace:distance:formula}
	transforms~\eqref{eq:lem:halfspace:zone:ineq} into
	\begin{equation*}
		\lVert p - q\rVert \dist(\{x_i(t)\},H)
		\geq
		\alpha \lVert p - q\rVert \dist(P_{M(t)}(t-1),H)
		\enspace.
	\end{equation*}
	Dividing both sides of the inequality by $\lVert p - q\rVert$ and taking the
	minimum over all $i\in [n]$ concludes the proof.
\end{proof}

\subsection{Convergence to Zero Volume}
\label{sec:converge_zero_vol}

As a first main result, we show that being $d$-broadcastable is enough for a message adversary to reduce the dimensionality of process values in $\IR^d$.
We also establish a rate for the contraction to a lower dimension.

\begin{theorem}\label{thm:poss}
	Let dimension $d \geq 1$.
	Every averaging algorithm with minimum broadcasting weight $\alpha > 0$ solves $d$-to-$(d-1)$-dimensional asymptotic subspace consensus in a $d$-broadcastable oblivious message adversary.
	Moreover, we have $\vol(P(t)) \leq \varepsilon$ if $t \geq \alpha^{-d} \log \frac{\vol(P(0))}{\varepsilon}$.
\end{theorem}

The remainder of the section is devoted to the proof of the theorem: we show that the $d$-dimensional volume containing the process values converges to $0$.
It follows that the process values lie within a subspace of lower dimension.
The proof also establishes a contraction rate towards zero volume.
An outline of the proof strategy is given in Figure~\ref{fig:panel_2}.

We introduce notation for certain symmetric bodies.
Let $r:\IR \to \IR_+$ be a measurable function.
Then $\Sym(r) \in \IR^d$ is defined as the (full) body obtained by rotating
function $r$ around the first coordinate axis.
Formally, we define:
For a first-axis value $\xi \in \IR$,
write $H_\xi$ for the $(d-1)$-dimensional hyperplane defined by normal vector $(1,0,\dots,0)$ and point $(\xi,0,\dots,0)$.
Let $\Ball_\xi(\rho) \subseteq H_\xi$ be the $(d-1)$-dimensional ball with center $(\xi, 0, \dots, 0)$ and radius $\rho$.
Then $\Sym(r) = \bigcup_{\xi \in \IR}\Ball_\xi(r(\xi))$.

Lemma~\ref{lem:volumes} states lower and upper bounds on the volumes of a partitioning
of such bodies for concave functions $r$.
We start with some technical results, Lemmas~\ref{lem:concave:line} and \ref{lem:concave:line:zero}, (Figure~\ref{fig:panel_2}b--c), which are affine upper and lower bounds to one-dimensional concave functions.
We will later use these to bound the radius function~$r$.

\begin{figure}[htp]
	\centering
	\includegraphics[width=\linewidth]{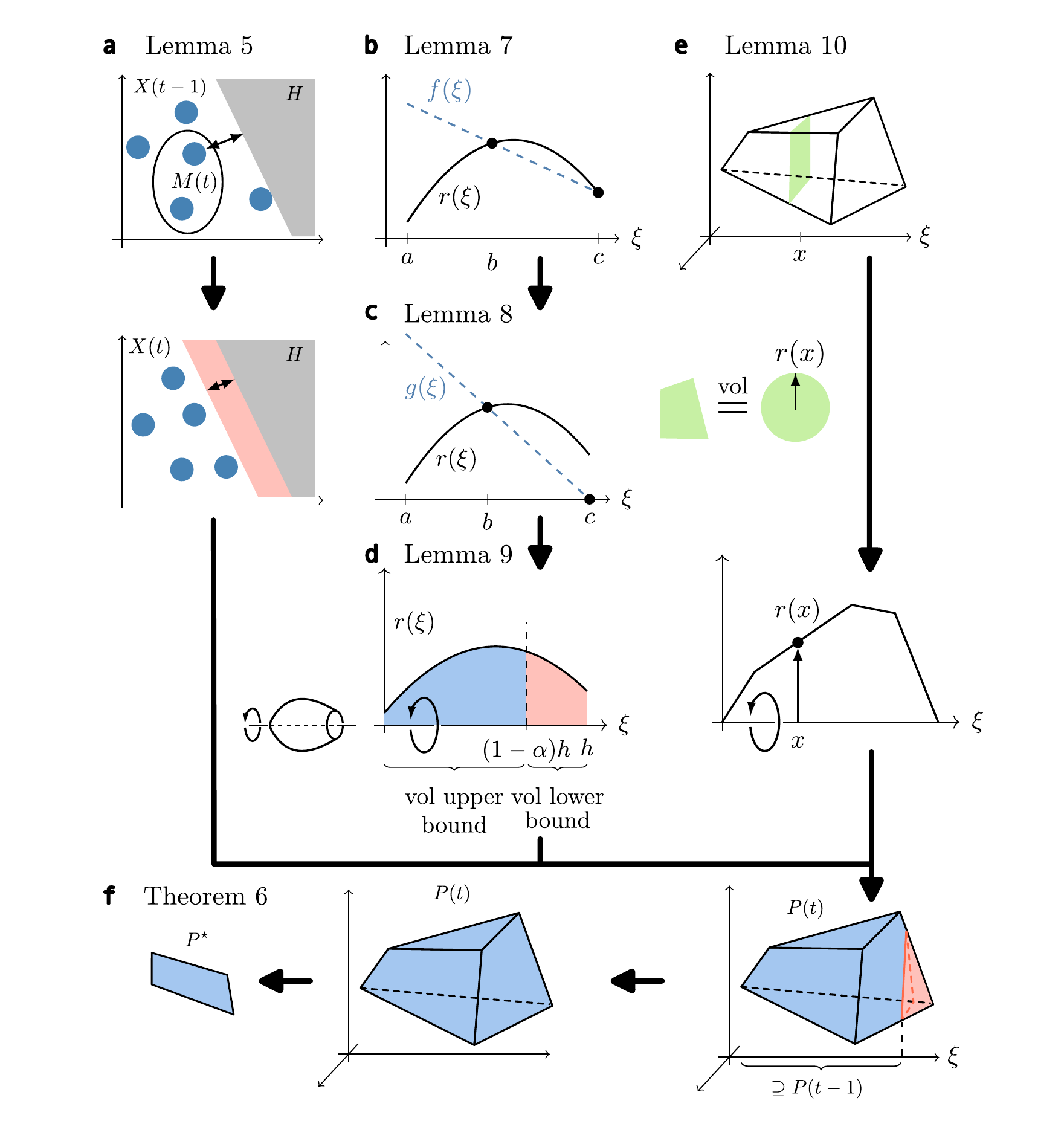}
	\caption{\textbf{Proof strategy for Theorem~\ref{thm:poss}.} \textbf{a} Attracting process values away from a hyperplane $H$. \textbf{b}-\textbf{c} Affine upper and lower bounds for concave radius functions. \textbf{d} Volume upper and lower bounds for segments of symmetric rotational objects. \textbf{e} Steiner-type symmetrization of object to rotational object. Volumes of cuts (green) and segments along the rotational axis are preserved. The resulting radius function is concave. Contraction of the volume in one round is shown. \textbf{f} Final assembly into convergence of the polyhedron of process values $P(t)$ to $P^\star$ which lies on a hyperplane.}
	\label{fig:panel_2}
\end{figure}

The following two lemmas on concave functions are proved in the Appendix.

\begin{lemmarep}[\leanproof{\leanurl{LemConcaveLine.lean}}]\label{lem:concave:line}
	Let $a,b,c \in \IR$ with $a<b<c$
	and let $r : [a,c]\to \IR$ be a concave function.
	Define $f:\IR \to \IR$ to be the unique affine function with $f(b) = r(b)$ and
	$f(c) = r(c)$.
	Then $f(\xi) \geq r(\xi)$ for all $\xi \in [a,b]$ and $f(\xi) \leq r(\xi)$ for
	all $\xi \in [b,c]$.
\end{lemmarep}
\begin{proof}
	It is
	\begin{equation}
		f(\xi)
		=
		\frac{r(b)-r(c)}{c-b} (c- \xi) + r(c)
	\end{equation}
	for all $\xi \in \IR$.

	By concavity of~$r$, whenever
	$\xi \in [a, b]$,
	we have
	$\alpha = (b - \xi)/(c - \xi) \in [0,1]$ and thus
	\begin{equation}
		\begin{split}
			r(b)
			=
			r\big((1-\alpha)\xi + \alpha c\big)
			\geq
			(1-\alpha)r(\xi) + \alpha r(c)
			=
			\frac{c-b}{c-\xi}\big(r(\xi) - r(c)\big) + r(c)
			\enspace,
		\end{split}
	\end{equation}
	which implies
	$f(\xi) \geq r(\xi)$ and proves the first part of the lemma.

	By concavity of~$r$,
	whenever
	$\xi \in [b, c]$,
	we have
	$\beta = (\xi-b)/(c-b) \in [0,1]$
	and thus
	\begin{equation}
		r(\xi)
		=
		r\big( (1-\beta)b + \beta c \big)
		\geq
		(1-\beta) r(b) + \beta r(c)
		=
		f(\xi)
		\enspace,
	\end{equation}
	which proves the second part of the lemma.
\end{proof}

We will actually only need a weaker linear bound (Figure~\ref{fig:panel_2}c) for our purposes:

\begin{lemmarep}[\leanproof{\leanurl{LemConcaveLineZero.lean}}]\label{lem:concave:line:zero}
	Let $a,b,c \in \IR$ with $a<b<c$
	and let $r : [a,c]\to \IR$ be a nonnegative concave function.
	Define $g:\IR \to \IR$ to be the unique affine function with $g(b) = r(b)$ and
	$g(c) = 0$.
	Then $g(\xi) \geq r(\xi)$ for all $\xi \in [a,b]$ and $g(\xi) \leq r(\xi)$ for
	all $\xi \in [b,c]$.
\end{lemmarep}
\begin{proof}
	Using Lemma~\ref{lem:concave:line}, it is sufficient to show that $g(\xi) \geq
		f(\xi)$ for all $\xi\leq b$ and $g(\xi) \leq f(\xi)$ for all $\xi\geq b$
	where~$f$ is the function defined by
	\begin{equation}
		f(\xi)
		=
		\frac{r(b)-r(c)}{c-b} (c- \xi) + r(c)
		\enspace.
	\end{equation}

	We note that
	\begin{equation}
		g(\xi)
		=
		\frac{r(b)}{c-b} (c- \xi)
	\end{equation}
	for all $\xi\in \IR$.

	It is $f(b) = g(b)$.
	The difference of the slopes of~$f$ and~$g$ is equal to $-r(c)/(c-b)$, which is
	nonpositive.
	This shows the claimed inequalities.
\end{proof}

This linear bound allows us to show that any $\alpha$-quantile of a symmetric
body on the first axis (red in Figure~\ref{fig:panel_2}d) contains a non-negligible portion of the total volume.
The proof is given in the Appendix.

\begin{lemmarep}[\leanproof{\leanurl{LemVolumes.lean}}]\label{lem:volumes}
	Let $h\in\IR_+$, let $\alpha \in (0,1)$, and let $r:\IR \to \IR_+$ be a measurable function that is concave in $[0,h]$.
	Define $r_\mathrm{left}, r_\mathrm{right}:\IR\to \IR_+$ by
	\[
		\begin{aligned}
			r_{\mathrm{left}}(\xi)  & =
			\begin{cases}
				r(\xi) & \text{if } \xi \in [0,(1-\alpha)h]\,, \\
				0      & \text{otherwise}
			\end{cases}
			\qquad
			r_{\mathrm{right}}(\xi) & =
			\begin{cases}
				r(\xi) & \text{if } \xi \in [(1-\alpha)h,h]\,, \\
				0      & \text{otherwise}
			\end{cases}
		\end{aligned}
	\]
	Setting $r_0 = r\big((1-\alpha)h\big)$ and $C_{d-1} = \pi^{(d-1)/2}/\Gamma(\frac{d-1}{2}+1)$,
	we have
	\begin{equation*}
		\vol(\Sym(r_\mathrm{left})) = \int_0^{(1-\alpha)h} C_{d-1} r(\xi)^{d-1} \ d\xi
		\leq
		C_{d-1} \frac{r_0^{d-1}h}{\alpha^{d-1}d} \left(1 - \alpha^d\right)
	\end{equation*}
	and
	\begin{equation*}
		\vol(\Sym(r_\mathrm{right})) = \int_{(1-\alpha)h}^{h} C_{d-1} r(\xi)^{d-1} \ d\xi
		\geq
		C_{d-1}\frac{r_0^{d-1}h}{\alpha^{d-1}d} \alpha^d
		\enspace.
	\end{equation*}
\end{lemmarep}
\begin{proof}
	For the volume of a symmetric body $\Sym(r)$ with radius function $r$,
	\[
		\vol_d(\Sym(r)) = \int_{-\infty}^{\infty}\vol_{d-1}\left( \Ball_\xi(r(\xi)) \right) \ d\xi\enspace,
	\]
	where the $(d-1)$-dimensional volume of the balls are given by
	\begin{equation}
		\vol_{d-1}\left( \Ball_\xi(r(\xi)) \right) = C_{d-1} \cdot r(\xi)^{d-1} \enspace. \label{eq:vol:b}
	\end{equation}
	Using the linear function~$g$ from Lemma~\ref{lem:concave:line:zero}, we have
	\begin{equation*}
		\begin{split}
			\int_0^{(1-\alpha)h} r(\xi)^{d-1} \ dx
			& \leq
			\int_0^{(1-\alpha)h} g(\xi)^{d-1} \ dx
			=
			\int_0^{(1-\alpha)h} \frac{r_0^{d-1}}{\alpha^{d-1}h^{d-1}} (h-\xi)^{d-1} \ d\xi
			\\ & =
			\frac{r_0^{d-1}}{\alpha^{d-1}h^{d-1}d} \big(h^d - \alpha^d h^d\big)
			=
			\frac{r_0^{d-1}h}{\alpha^{d-1}d} \left(1 - \alpha^d\right)
		\end{split}
	\end{equation*}
	and
	\begin{equation*}
		\begin{split}
			\int_{(1-\alpha)h}^h r(\xi)^{d-1} \ dx
			& \geq
			\int_{(1-\alpha)h}^h g(\xi)^{d-1} \ dx
			=
			\int_{(1-\alpha)h}^h \frac{r_0^{d-1}}{\alpha^{d-1}h^{d-1}} (h-\xi)^{d-1} \ d\xi
			\\ & =
			\frac{r_0^{d-1}}{\alpha^{d-1}h^{d-1}d} \alpha^d h^d
			=
			\frac{r_0^{d-1}h}{\alpha^{d-1}d} \alpha^d
			\enspace,
		\end{split}
	\end{equation*}
	which concludes the proof.
\end{proof}

We can now show contraction of the volume of $P(t)$.
We will do this by applying a Steiner-type symmetrization to the polyhedron $P(t) \subseteq \IR^d$
obtaining a body $\Sym(r) \in \IR^d$, with some radius function $r$ (see Figure~\ref{fig:panel_2}e).
The body $\Sym(r)$ is obtained by rotating function $r$ around the first coordinate axis.
By construction, the cuts $P(t) \cap H_x$ and $\Sym(r) \cap H_x$, the latter of which is the ball $\Ball_x(r(x))$, have the same $(d-1)$-dimensional volume (green in the figure).
Crucial in the proof is showing that the function $r$ is concave, so that Lemmas~\ref{lem:concave:line} and~\ref{lem:concave:line:zero} can be applied to lower bound the volume shaped off the symmetric body.

\begin{lemma}[\leanformalized{\leanurl{LemContract.lean}}]\label{lem:contract}
	If for round $t \geq 1$, broadcasting set~$M(t)$ fulfills $\lvert M(t)\rvert \leq d$,
	then
	\begin{equation}
		\frac{\vol(P(t))}{\vol(P(t-1))} \leq 1-\alpha^d
		\enspace.
	\end{equation}
\end{lemma}
\begin{proof}
	For brevity, set $P = P(t-1)$ and $P' = P(t)$.
	Let $A\subseteq \IR^d$ be a hyperplane that contains all values of processes of the broadcasting set~$M(t)$.
	Such a hyperplane exists because $\lvert M(t)\rvert \leq d$.
	Without loss of generality, by rotating and translating the coordinate system accordingly,
	we assume that $A = \{z\in\IR^d \mid z_1 = 0\}$.

	\medskip

	\noindent\emph{Symmetrization of output polyhedron.}
	We start with a Steiner-type symmetrization of~$P$ along the first axis.
	That is, we choose the radius function $r:\IR\to\IR_+$ such that
	\begin{align}
		\vol_{d-1}(H_x \cap P) = \vol_{d-1}(H_x \cap \Sym(r)) = \vol_{d-1}(B_x(r(x))) = C_{d-1} \cdot r(x)^{d-1}\label{eq:volume:ball}
	\end{align}
	the latter of which follows from \eqref{eq:vol:b} with $C_{d-1} = \pi^{(d-1)/2}/\Gamma\left(\frac{d-1}{2}+1\right)$.
	One observes that any sections of $P$ and $\Sym(r)$ within some $[a,b]$ along the first axis
	have same volumes:
	\[
		\int_{a}^{b}\vol_{d-1}\left( H_\xi \cap P \right)\ d\xi =
		\int_{a}^{b}\vol_{d-1}\left( H_\xi \cap \Sym(r) \right)\ d\xi
		\enspace.
	\]
	In particular, the total volume of both objects is the same.

	We next show that the radius function $r$ is concave, following the proof
	by Charron-Bost, Függer, and Nowak~\cite{CBFN16:centroid}.
	First, observe that $r$ is zero outside some finite interval $[a,b]$
	by the fact that the polyhedron $P$ is finite.
	We will show that $r$ is concave in $[a,b]$.
	Let $t \in [0,1]$ and $x,y \in [a,b]$, and abbreviate $S =\Sym(r)$.
	By definition of the Minkowski sum,
	\[
		t \left(H_x \cap P\right) + (1-t) \left(H_y \cap P\right)
		= \left\{ t u + (1-t) v \mid u \in H_x \cap P \wedge v \in H_y \cap P \right\}
	\]
	and from the fact that $P$ is convex, for any $u \in H_x \cap P$ and $v \in H_y \cap P$,
	it is $t u + (1-t) v \in P \cap H_{tx + (1-t)y}$, and thus,
	\[
		t \left(H_x \cap P\right) + (1-t) \left(H_y \cap P\right) \subseteq H_{tx + (1-t)y} \cap P \enspace.
	\]
	Consequently,
	\begin{align*}
		\vol_{d-1}\left( H_{tx + (1-t)y} \cap P \right)^{\frac{1}{d-1}}
		 & \geq \vol_{d-1}\left( t \left(H_x \cap P\right)^{\frac{1}{d-1}} +
		(1-t) \left(H_y \cap P\right) \right)^{\frac{1}{d-1}}\enspace.
	\end{align*}
	Applying the Brunn-Minkowski inequality
	\[
		\vol_{d-1}\left( t(H_x \cap P) + (1-t)(H_y \cap P) \right)^{\frac{1}{d-1}}
		\geq t \vol_{d-1}\left( H_x \cap P \right)^{\frac{1}{d-1}} +
		(1-t) \vol_{d-1}\left( H_y \cap P \right)^{\frac{1}{d-1}}\\
	\]
	to the right side yields
	\begin{align*}
		\vol_{d-1}\left( H_{tx + (1-t)y} \cap P \right)^{\frac{1}{d-1}}
		                          & \geq t\vol_{d-1}\left( H_x \cap P \right)^{\frac{1}{d-1}} +
		(1-t)\vol_{d-1}\left( H_y \cap P \right)^{\frac{1}{d-1}}\enspace.                       \\
		\intertext{By construction of the symmetrization,}
		\vol_{d-1}\left( H_{tx + (1-t)y} \cap S \right)^{\frac{1}{d-1}}
		                          & \geq t\vol_{d-1}\left( H_x \cap S \right)^{\frac{1}{d-1}} +
		(1-t) \vol_{d-1}\left( H_y \cap S \right)^{\frac{1}{d-1}}\enspace.
		\intertext{Combining with \eqref{eq:volume:ball}, one obtains for the radius function $r$:}
		r\left(tx + (1-t)y\right) & \geq t \cdot r(x) + (1-t) \cdot r(y)
	\end{align*}
	from which concavity of $r$ in $[a,b]$ follows.

	\medskip

	\noindent\emph{Volume bounds.}
	Next, define the half-spaces
	$A_+ = \{z\in \IR^d \mid z_1 > 0\}$
	and
	$A_- = \{z\in \IR^d \mid z_1 < 0\}$.
	It is sufficient to show
	\begin{equation}\label{eq:contract:in:halfspace}
		{\vol(P'\cap A_+)}
		\leq
		(1-\alpha^d)
		{\vol(P\cap A_+)}
	\end{equation}
	for then
	\begin{equation*}
		\vol(P')
		=
		{\vol(P'\cap A_+)}
		+
		{\vol(P'\cap A_-)}
		\leq
		(1-\alpha^d)
		\big(
		{\vol(P\cap A_+)}
		+
		{\vol(P\cap A_-)}
		\big)
		=
		(1-\alpha^d)
		{\vol(P)}
	\end{equation*}
	by symmetry.

	Define $h = \sup\{ z_1 \mid z\in P\}$.
	The function~$r$ is concave on the interval~$[0,h]$.
	Lemma~\ref{lem:halfspace:zone} applied to the open half-space
	$H = \{ z\in \IR^d \mid z_1 > h\}$, we see that
	$P' \cap (1-\alpha)H = \emptyset$.
	Thus, by the non-expansion property of averaging algorithms, we have
	\begin{equation*}
		\vol(P'\cap A_+) \leq \vol\big((P\cap A_+) \setminus (1-\alpha)H\big) =  \vol(\Sym(r_\mathrm{left}))
	\end{equation*}
	where~$r_\mathrm{left}$ and~$r_\mathrm{right}$ are defined as in Lemma~\ref{lem:volumes}.
	We thus conclude
	\begin{equation*}
		\begin{split}
			\frac{\vol(P'\cap A_+)}{\vol(P\cap A_+)}
			& \leq
			\frac{\vol(\Sym(r_\mathrm{left}))}{\vol(\Sym(r_\mathrm{left})) + \vol(\Sym(r_\mathrm{right}))}
			=
			\frac{1}{1 + \vol(\Sym(r_\mathrm{right}))/ \vol(\Sym(r_\mathrm{left}))}
			\\
			& \leq
			\frac{1}{1 + \alpha^d/(1-\alpha^d)}
			=
			1 - \alpha^d
			\enspace,
		\end{split}
	\end{equation*}
	which proves~\eqref{eq:contract:in:halfspace}
	and thus the lemma.
\end{proof}

Lemma~\ref{lem:contract} provides everything to establish contraction of the volume to zero.
To allow us to conclude that this indeed implies a lower dimensionality, we use the following lemma, whose proof is given in the Appendix.

\begin{lemmarep}[\leanproof{\leanurl{LemConvexZeroHyperplane.lean}}]\label{lem:convex:zero:hyperplane}
	Let $C \subseteq \mathbb{R}^d$ be a convex set.
	The following are equivalent:
	\begin{enumerate}
		\item There exists an affine hyperplane $E \subseteq \mathbb{R}^d$ such that $C \subseteq E$.
		\item $\vol(C) = 0$.
	\end{enumerate}
\end{lemmarep}
\begin{proof}
	($\Rightarrow$)
	If $C \subseteq E$ for some affine hyperplane $E$, then $E$ has Lebesgue measure zero in $\mathbb{R}^d$.
	Hence $\vol(C) = 0$.

	\smallskip

	($\Leftarrow$)
	By means of contradiction, assume that $C$ is not contained in any affine hyperplane.
	Then there exist $d+1$ affinely independent points $x_0,\dots,x_d \in C$, i.e.,
	$d$ vectors $x_1 - x_0,\dots,x_d-x_0$ that are linearly independent.
	Their convex hull
	$
		\poly(\{x_0,\dots,x_d\})
	$ is a simplex and has volume $\vol(\poly(\{x_0,\dots,x_d\})) = |\det([x_1 - x_0,\dots,x_d-x_0])|/d!$.
	From the linear independence we have $|\det([x_1 - x_0,\dots,x_d-x_0])| > 0$
	from which it follows that $\vol(\poly(\{x_0,\dots,x_d\})) > 0$.
	Since $C$ is convex and contains all points $x_i$, we have $\poly(\{x_0,\dots,x_d\}) \subseteq C$.
	Consequently,
	$
		\vol(C) \ge \vol(\poly(\{x_0,\dots,x_d\})) > 0,
	$
	which contradicts the assumption $\vol(C)=0$.
\end{proof}

\begin{proof}[Proof of Theorem~\ref{thm:poss}]
	We are now in the position to assemble everything to prove Theorem~\ref{thm:poss}.
	By Lemma~\ref{lem:contract}, we have $\vol(P(t)) \to 0$ as $t\to\infty$.
	Since the sets~$P(t)$ are non-increasing for any averaging algorithm, this is equivalent to the limit set $P^* = \lim_{t\to\infty} P(t) = \bigcap_{t\geq 0} P(t)$ being contained in some affine hyperplane $E\subseteq \IR^d$.

	The set~$P^*$ is convex as the intersection of the convex sets~$P(t)$.
	By Lemma~\ref{lem:contract}, we have
	\begin{equation*}
		\vol(P^*) = \lim_{t\to\infty} \vol(P(t)) \leq \lim_{t\to\infty} \left(1 - \alpha^d \right)^t \cdot \vol(P(0)) = 0
		\enspace.
	\end{equation*}
	Lemma~\ref{lem:convex:zero:hyperplane} now shows the existence of the hyperplane~$E$ such that $P^*\subseteq E$.
\end{proof}

From the reductions of $d$-rooted to $d$-broadcastable oblivious message adversaries (Theorem~\ref{thm:rooted:to:broadcastable}),
and the speed-up construction of alternating sufficiently long consecutive rounds of relaying messages with an averaging round \cite{CBFN16:icalp}, one obtains:

\begin{corollary}\label{cor:bound}
	Let $0\leq \alpha < 1$.
	In a $d$-rooted oblivious message adversary, every $\alpha$-safe averaging algorithm solves $d$-to-$(d-1)$-dimensional asymptotic subspace consensus.
	Moreover, alternating $\frac{\pi^2+6}{6}n + 1$ rounds of relaying messages with one round of an $\alpha$-safe averaging algorithm, solves $d$-to-$(d-1)$-dimensional asymptotic subspace consensus.
	Moreover, we have $\vol(P(t)) \leq \varepsilon$ if $t \geq \left(\frac{\pi^2+6}{6}n + 1\right) \alpha^{-d} \log \frac{\vol(P(0))}{\varepsilon}$.
\end{corollary}

For example, by observing that the equal neighbor averaging algorithm is $1/n$-safe, Corollary~\ref{cor:bound} holds with $\alpha = 1/n$ for this algorithm with a time in $O\left(n^{d+1} \log \frac{\vol(P(0))}{\varepsilon}\right)$.

\subsection{Convergence to Lower-dimensional Subspace}
\label{sec:converge_lowdim}

Theorem~\ref{thm:poss} showed a contraction to lower dimensionality than the initial dimensionality~$d$.
In this section we show that the dimensionality reduction indeed matches the lower bound in Section~\ref{sec:problem}:

\begin{theorem}[\leanproof{\leanurl{ThmSolving.lean}}]\label{thm:solving}
	Let $d \geq k \geq 1$.
	Every averaging algorithm with minimum broadcasting weight $\alpha > 0$ solves $d$-to-$(k-1)$-dimensional asymptotic subspace consensus in a $k$-broadcastable oblivious message adversary.
\end{theorem}

\begin{figure}[bth]
	\centering
	\includegraphics[width=\linewidth]{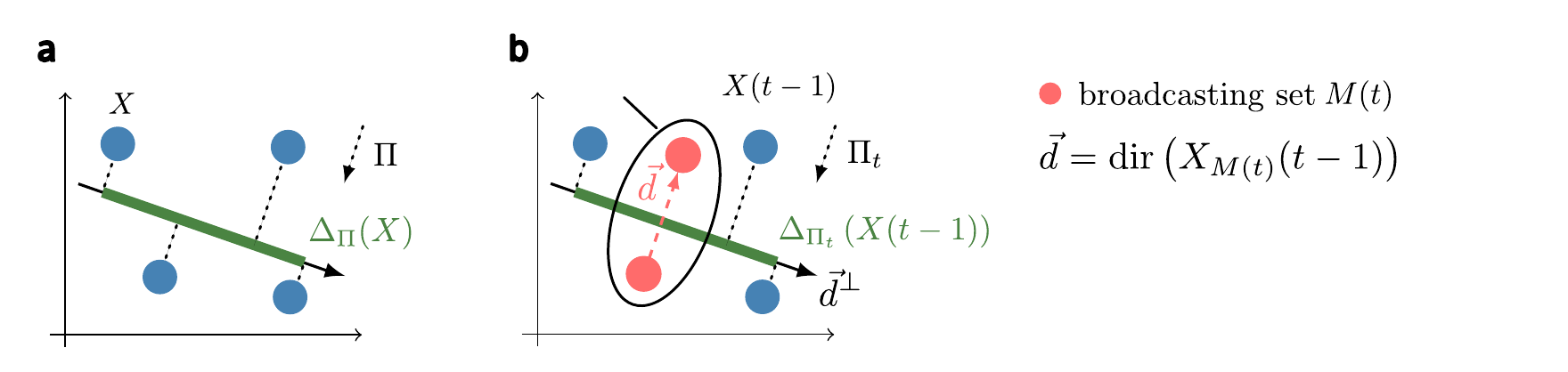}
	\caption{\textbf{Thickness $\Delta$ with respect to projection.} \textbf{a} Points $X \subseteq \IR^d$ (in blue) are projected orthogonally via $\Pi$ (dashed). The thickness is the diameter (green) of the projected image.
		\textbf{b} Same as in \textbf{a}, for the case $X = X(t-1)$ and $\Pi$ is chosen as $\Pi_t$, i.e., in parallel to the direction ($\vec{d}$) of the outputs in round $t-1$ of those processes in the broadcasting $M(t)$.
		Here, $|M(t)| = 2$ processes, with outputs in red).}
	\label{fig:panel_3}
\end{figure}

The remainder of the section is devoted to showing this results.
The proof idea is to decompose process outputs (Lemma~\ref{lem:xi_xip}) and differences of process outputs (Lemma~\ref{lem:differences}) into components that are within or parallel to the polyhedron of process outputs of the broadcasting set, as well as a remainder.
We then define a so-called thickness with respect to an orthogonal projection (Figure~\ref{fig:panel_3}a) and orthogonal projections onto the linear space spanned by the process outputs in the broadcasting set in Definition~\ref{def:thickness}.
We combine these to follow the dynamics of the thickness with respect to the latter projection (Figure~\ref{fig:panel_3}b) during an execution.

We then show (Lemmas~\ref{lem:accumulation_point} and~\ref{lem:continuity}) that the orthogonal projections constructed from the broadcasting set have an accumulation point and the thickness is continuous in the projections, allowing us to measure contraction of the thickness with respect to this accumulation point (Lemma~\ref{lem:Delta-contraction}).

We start with the decomposition of node values into convex combinations of values from the broadcasting set
and all node values; the proof is given in the Appendix.

\begin{lemmarep}[\leanproof{\leanurl{LemXiXip.lean}}]\label{lem:xi_xip}
	Assume an averaging algorithm with minimum broadcasting weight $\alpha > 0$.
	For any $i \in [n]$ and round $t \geq 1$, there exist points $\xi(t-1) \in P_{M(t)}(t-1)$  and $\xi'(t-1) \in P(t-1)$, such that
	$
		x_i(t) = \alpha \cdot \xi(t-1) + (1 - \alpha) \cdot \xi'(t-1)
	$.
\end{lemmarep}
\begin{proof}
	Let $t \geq 1$.
	Let $w_M = \sum_{j \in M(t)} w_{ij}(t)$.
	From the algorithm, we have
	\begin{align*}
		x_i(t) & = \sum_{j \in [n]} w_{ij}(t) x_j(t-1)                                                       \\
		       & = \sum_{j \in [n] \setminus M(t)} w_{ij}(t) x_j(t-1) + \sum_{j \in M(t)} w_{ij}(t) x_j(t-1)
	\end{align*}
	We can rewrite
	\begin{align*}
		\sum_{j \in M(t)} w_{ij}(t) x_j(t-1) & = w_M \sum_{j \in M(t)} \frac{w_{ij}(t)}{w_M} x_j(t-1)
		= w_M \cdot \xi
	\end{align*}
	for some point $\xi$ in the affine root plane by the fact that $\sum_{j \in M(t)}\frac{w_{ij}(t)}{w_M} = 1$.
	If $w_M = 1$, we set $\hat{\xi} = \xi$.
	Otherwise, we have
	\begin{align*}
		\sum_{j \in [n] \setminus M(t)} w_{ij}(t) x_j(t-1) & = (1 - w_M) \sum_{j \in [n] \setminus M(t)} \frac{w_{ij}(t)}{1 - w_M} x_j(t-1)
		= (1 - w_M) \cdot \hat{\xi}
	\end{align*}
	for some convex combination $\hat{\xi}$ of points $x_j(t-1)$.
	From this, we have that
	$$
		x_i(t) = w_M \cdot \xi + (1 - w_M) \cdot \hat{\xi} \enspace
	$$
	which we can rewrite as
	$$
		x_i(t) = \alpha \cdot \xi + (1-\alpha)\left(
		\frac{w_M - \alpha}{1 - \alpha} \cdot \xi + \frac{1 - w_M}{1-\alpha} \cdot \hat{\xi}
		\right)\enspace.
	$$
	The right term is a convex combination of points $\xi$ and $\hat{\xi}$, both of which are a convex combination of points $x_j(t-1)$.
	Thus it is a convex combination of points $x_j(t-1)$.
	The lemma's statement follow from setting $\xi(t-1) = \xi$ and $\xi'(t-1) = \frac{w_M - \alpha}{1 - \alpha} \cdot \xi + \frac{1 - w_M}{1-\alpha} \cdot \hat{\xi}$.
\end{proof}

If not stated otherwise, we consider the inner product space $\IR^d$ with the inner product $\bra \cdot,\cdot \ket$.
For a set $X \subseteq \IR^d$, we write $\linspan(X)$ for the span of $X$, i.e., the linear sub-space of $\IR^d$ of all linear combinations of vectors in $X$.
For a set $X$ of points, we write $\dir(X)$ for the direction space $\linspan(\{x - y \mid x,y \in X\})$.
Of particular interest will be $\dir(X_{M(t)}(t-1))$ which has some dimensionality $m \leq k-1$.
For a sub-space $X$ of the inner product space $\IR^d$, we write $X^\bot$ for the orthogonal complement
of $X$, i.e., the set of vectors $v \in \IR^d$ that are orthogonal ($\bra v, v' \ket = 0$) to all vectors $v' \in X$.

Using the decomposition for two node values in Lemma~\ref{lem:xi_xip}, we can also decompose their difference as shown in the following lemma.

\begin{lemma}[\leanproof{\leanurl{LemDifferences.lean}}]\label{lem:differences}
	Assume an algorithm with minimum broadcasting weight $\alpha > 0$.
	For any $i,j \in [n]$ and $t \geq 1$, if the broadcasting set fulfills $|M(t)| \geq 1$, there exists a vector $u_{\parallel}(t-1) \in \dir(X_{M(t)}(t-1))$, and a residual vector $u_\text{res}(t-1) \in \{x-y \mid x,y \in P(t-1)\}$, such that
	$
		x_i(t)-x_j(t) = \alpha \cdot u_{\parallel}(t-1) + (1 - \alpha) \cdot u_\text{res}(t-1)
	$.
\end{lemma}
\begin{proof}
	Follows from the fact that for $x,y \in X \subseteq \IR^d$, it is $x-y \in \dir(X)$ and
	application of Lemma~\ref{lem:xi_xip} to both points in the difference $x_i(t)-x_j(t)$.
\end{proof}

\begin{definition}[Projection $\Pi_t$ and Thickness $\Delta$\leanformalized{\leanurl{DefThickness.lean}}]\label{def:thickness}
	Consider an averaging algorithm with process values in $\IR^d$.
	For a round $t\ge 1$ with broadcasting set $M(t) \neq \emptyset$ let
	$
		\Pi_t : \mathbb{R}^d \to \mathbb{R}^d
	$
	denote the orthogonal projection onto $\dir(X_{M(t)}(t-1))^\perp$.

	Let $X \subseteq \IR^d$ be a non-empty finite set.
	Let $\Pi: \IR^d \to \IR^d$ be an orthogonal projection.
	Define the \emph{thickness of $X$ under $\Pi$} by
	$
		\Delta_\Pi(X) := \max_{x, y \in X} \| \Pi ( x - y ) \|
	$.
\end{definition}
See also Figure~\ref{fig:panel_3} for a visualization.
From the fact that $\Pi$ is an orthogonal projection, we have $\| \Pi( x - y ) \| \leq \| x - y \|$.

\begin{lemma}[\leanproof{\leanurl{LemAccumulationPoint.lean}}]\label{lem:accumulation_point}
	Assume an averaging algorithm executed in a $k$-broadcastable oblivious message adversary for some $k \geq 1$.
	Then the sequence $(\Pi_t)_{t \geq 1}$ has an accumulation point that is an orthogonal projection in $\IR^d$.
\end{lemma}
\begin{proof}
	The set $\cal P$ of orthogonal projections in $\IR^d$ is identified with the set of matrices $\{ \Pi \in \IR^{d \times d} \mid \Pi^2 = \Pi = \Pi^\top \}$.
	For any $\Pi \in {\cal P}$, eigenvalues are in $\{0, 1\}$ and thus $\|\Pi\|_\text{op} \in \{0, 1\}$.
	It follows that $\cal P$ is bounded.
	Further, if a sequence $(Q_n)_{n \geq 0}$ with $Q_n \in {\cal P}$ has limit $Q^\star$, then ${Q^\star}^2 = {Q^\star} = {Q^\star}^\top$, from which closedness of $\cal P$ follows.
	Thus $\cal P$ is compact.

	From the compactness of $\cal P$, it follows that the sequence $(\Pi_t)_{t \geq 1}$, with each $\Pi_t \in {\cal P}$, has an accumulation point in ${\cal P}$.
\end{proof}

The following lemma is proved in the Appendix.

\begin{lemmarep}[\leanproof{\leanurl{LemContinuity.lean}}]\label{lem:continuity}
	For any finite non-empty set $X \subseteq \IR^d$ and with ${\cal P}$ being the set of orthogonal projections in $\IR^d$, the function $\Pi \mapsto \Delta_\Pi(X) : {\cal P} \to \IR$ is Lipschitz continuous.
\end{lemmarep}
\begin{proof}
	Let $\Pi,Q \in {\cal P}$ and $x,y \in X$.
	Then from the inverse triangle inequality and the definition of the operator norm,
	\[
		| \|\Pi(x - y)\| - \|Q(x-y)\| | \leq \| (\Pi-Q)(x-y) \| \leq \|(\Pi-Q)\|_\text{op} \cdot \|x-y\|\enspace.
	\]
	Maximizing over all $x,y \in X$, yields
	\[
		|\Delta_\Pi(X) - \Delta_Q(X)| \leq \|(\Pi-Q)\|_\text{op} \cdot \max_{x,y \in X}\|x-y\|\enspace;
	\]
	and thus Lipschitz continuity.
\end{proof}

\begin{lemma}[Contraction of thickness\leanproof{\leanurl{LemDeltaContraction.lean}}]
	\label{lem:Delta-contraction}
	Let $k \geq 1$.
	Assume an averaging algorithm executed in a $k$-broadcastable oblivious message adversary with minimum broadcasting weight $\alpha>0$.
	Then there exists an orthogonal projection $\Pi: \IR^d \to \IR^d$ whose kernel has dimension at most $k-1$ such that $\lim_{t \to \infty} \Delta_\Pi(X(t)) = 0$.
\end{lemma}
\begin{proof}
	Let $\Pi$ be such an orthogonal projection in the following.
	We show the lemma in several steps:

	\medskip
	\noindent\emph{Monotonicity of thickness.}
	Since each process updates its value to a convex combination of received values, we have $P(t) \subseteq P(t-1)$ for $t \geq 1$.
	For any projection $\Pi$ and sets $A \subseteq B$, it is $\Delta_\Pi(A) \leq \Delta_\Pi(B)$.
	It follows that for all $t \geq 1$,
	\begin{equation}
		\Delta_\Pi(P(t)) \leq \Delta_\Pi(P(t-1)) \enspace. \label{eq:monotonicity}
	\end{equation}

	\medskip
	\noindent\emph{Contraction of thickness in one round.}
	By Lemma~\ref{lem:differences} applied at time $t \geq 1$, for every $i,j \in [n]$ there exist
	$u_\parallel(t-1) \in \dir(X_{M(t)}(t-1))$ and
	$u_{\mathrm{res}}(t-1) \in \{x - y \mid x,y \in P(t-1)\}$ such that
	\[
		x_i(t)-x_j(t)
		= \alpha\,u_\parallel(t-1) + (1-\alpha)\,u_{\mathrm{res}}(t-1) \enspace.
	\]
	Apply the projection $\Pi_t$ that is orthogonal onto $\dir(X_{M(t)}(t-1))^\perp$, it is
	\[
		\Pi_t(x_i(t)-x_j(t))
		= (1-\alpha)\,\Pi_t(u_{\mathrm{res}}(t-1)),
	\]
	since $\Pi_t(u_\parallel(t-1)) = 0$ by construction of the projection.
	Taking norms, yields
	\[
		\|\Pi_t(x_i(t)-x_j(t))\|
		\le (1-\alpha)\,\|\Pi_t(u_{\mathrm{res}}(t-1))\| \enspace.
	\]
	Combining with
	\[
		\|\Pi_t(u_{\mathrm{res}}(t-1))\| \leq \sup_{x,y \in P(t-1)} \|\Pi_t(x - y)\| \leq \max_{x,y \in X(t-1)} \|\Pi_t(x - y)\|
	\]
	where the last inequality is due to the distances in a polyhedron being maximized by corner points,
	we obtain that for all $t \geq 1$,
	\begin{equation}
		\Delta_{\Pi_t}(X(t)) \leq (1-\alpha) \Delta_{\Pi_t}(X(t-1)) \enspace. \label{eq:one-step-contract}
	\end{equation}

	\medskip
	\noindent\emph{Contraction for fixed projection.}
	By Lemma~\ref{lem:accumulation_point}, the sequence of projections $(\Pi_t)_{t\geq 1}$, each with kernel of dimension at most $k-1$, has an accumulation point $\Pi^\star$ that itself is an orthogonal projection with kernel of dimension at most $k-1$.

	Let $(\Pi_{\tau_i})_{i \geq 0}$ be a subsequence with limit $\Pi^\star$.
	By \eqref{eq:one-step-contract} and Lemma~\ref{lem:continuity}, for each $i \geq 1$ we have
	\[
		\Delta_{\Pi_{\tau_i}}(X(\tau_i))
		\le (1-\alpha)\Delta_{\Pi_{\tau_{i-1}}}(X(\tau_{i-1}))
		+ \varepsilon_i
	\]
	where we choose
	\[
		\varepsilon_i = (1-\alpha)\,
		\big|
		\Delta_{\Pi_{\tau_i}}(X(\tau_{i-1}))
		-
		\Delta_{\Pi_{\tau_{i-1}}}(X(\tau_{i-1}))
		\big| \enspace.
	\]
	By Lemma~\ref{lem:continuity} and boundedness of $X(t)$, there exists a constant $C>0$ such that
	\[
		\varepsilon_i \le C \,\|\Pi_{\tau_i}-\Pi_{\tau_{i-1}}\|_{\mathrm{op}} \enspace.
	\]
	Since $\Pi_{\tau_i} \to \Pi^\star$, it follows that $\varepsilon_i \to 0$.

	We next show that the recursion
	\[
		a_i \le (1-\alpha)a_{i-1} + \varepsilon_i
	\]
	with $a_i = \Delta_{\Pi_{\tau_i}}(X(\tau_i))$ implies
	$
		\lim_{i\to\infty} a_i = 0
	$.

	To see this, let $\delta>0$.
	Since $\varepsilon_i \to 0$, there exists a $K$ such that
	$\varepsilon_i \le \delta$ for all $i \ge K$.
	Unrolling the recursion yields
	\[
		a_i
		\le (1-\alpha)^i a_0
		+ \sum_{k=1}^i (1-\alpha)^{\,i-k}\varepsilon_k \enspace.
	\]
	Splitting the sum at $K$, we obtain
	\[
		a_i
		\le (1-\alpha)^i a_0
		+ \sum_{k=1}^{K-1}(1-\alpha)^{\,i-k}\varepsilon_k
		+ \delta \sum_{k=K}^{i}(1-\alpha)^{\,i-k} \enspace.
	\]
	The first two terms converge to $0$ as $i\to\infty$, and the last term is bounded by
	\[
		\delta \sum_{m=0}^{\infty}(1-\alpha)^m = \frac{\delta}{\alpha} \enspace.
	\]
	Hence $\limsup_{i\to\infty} a_i \le \delta/\alpha$.
	Since $\delta$ was arbitrary, it follows that $\lim_{i\to\infty} a_i = 0$.

	From Lemma~\ref{lem:continuity}, using Lipschitz continuity of $\Pi \to \Delta_\Pi(\cdot)$
	and the fact that $\Pi_{\tau_i} \to \Pi^\star$,
	\[
		\lim_{i \to \infty} \Delta_{\Pi^\star}(X(\tau_i)) = 0\enspace.
	\]
	Applying monotonicity from \eqref{eq:monotonicity} for rounds $t$ between successive $\tau_i$, one obtains convergence of the thickness along the rounds. Thus,
	\[
		\lim_{t \to \infty} \Delta_{\Pi^\star}(X(t)) = 0\enspace.
	\]
	The lemma's statement follows by setting $\Pi = \Pi^\star$ and observing that
	the kernel of $\Pi$ has dimension at most $k-1$.
\end{proof}

\begin{proof}[Proof of Theorem~\ref{thm:solving}]
	Let $d \geq k \geq 1$.
	For an averaging algorithm with initial values in $\IR^d$ and minimum broadcasting weight $\alpha > 0$ executed in a $k$-broadcastable oblivious message adversary, by Lemma~\ref{lem:Delta-contraction}
	there exists an orthogonal projection $\Pi$ with $\lim_{t\to\infty}\Delta_\Pi(X(t)) = 0$.
	From the fact that $\ker(\Pi)$ has dimension at most $k-1$, the sequence
	$(X(t))_t$ converges onto a $(k-1)$-dimensional affine subspace, fulfilling (Subspace Agreement) for the $d$-to-$(k-1)$-dimensional subspace consensus problem.
	By the fact that for averaging algorithms $P(t) \subseteq P(0)$, (Validity) is fulfilled.
	The theorem's statement follows.
\end{proof}

From Theorem~\ref{thm:solving}, the reduction of message adversaries in Theorem~\ref{thm:rooted:to:broadcastable}, and the lower bound in Theorem~\ref{lem:imposs}, one finally obtains a complete characterization in oblivious message adversaries:

\begin{corollary}[Characterization of Asymptotic Subspace Consensus]
	Let $d > s \geq 0$.
	The problem of $d$-to-$s$-dimensional asymptotic subspace consensus in an oblivious message adversary $\N$ is solvable, if and only if  $\N$ is $(s+1)$-rooted.
	Any $\alpha$-safe averaging algorithm with $\alpha > 0$, and in particular the equal neighbor algorithm with $\alpha = 1/n$, with and without bounded periods of intermediate message relaying rounds, is a solution.
\end{corollary}

\section{Conclusion}\label{sec:conclusion}

We showed that a weakening of rooted oblivious message adversaries to
$d$-rooted adversaries achieves a dimension-reduction below the initial dimension~$d$.
Central to the proof is a symmetrization of the convex hull of process outputs to a body that is symmetric around the first axis.
This allows one to lower bound the volume that is shaped off the body by contraction of the convex hull, establishing also a bound
on the convergence speed towards a lower dimensionality.

We then extend this analysis to a complete characterization of when a contraction to a subspace of dimension $s < d$ can be achieved, showing that this is precisely the case if the oblivious message adversary is $(s+1)$-rooted.
Moreover, simple averaging algorithms are shown to be solution to this problem.
This shows that averaging algorithms degrade gracefully in (periods of) non-1-rooted message adversaries and that averaging algorithms are
also effective if convergence to a single point is not necessary.

The work also raises several follow-up questions.
An open problem is the quantification of the speed to dimensions $s < d$.
We also hypothesize that variants of the mid-point algorithm that have been shown to lead to fast convergence in one-dimensional \cite{CBFN16:icalp} and multidimensional settings \cite{CBFN16:centroid,FN18:disc}, provide significantly faster convergence than the $\alpha$-safe algorithms shown in this work.
In particular we conjecture that analogous to asymptotic consensus in rooted adversaries, the dependency on the dimension~$d$ can be removed.

\bibliographystyle{plainurl}
\bibliography{agents}

\end{document}